\documentclass[submission,copyright,creativecommons]{eptcs}
\providecommand{\event}{TARK 2015} 

\usepackage{pgf}
\usepackage{tikz}
\usetikzlibrary{arrows}

\input{xy}
\xyoption{all} \xyoption{arc}

\usepackage{amsmath,amssymb,amsthm}

\newcommand{\Set}{set}

\newcommand{\Nat}{ {\omega} }
\newcommand{\Real}{ {\mathbb{R}} }

\newcommand\closure[1]{\overline{#1}}

\newtheorem{theorem}{Theorem}
\newtheorem{proposition}{Proposition}
\newtheorem{definition}{Definition}

\newtheorem{example}{Example}
\newtheorem{observation}{Observation}

\newcommand{\Prob}{ {\mathbb P} }
\newcommand{\ES}{ {\mathbb S} }

\newcommand{\Obs}{{\mathcal{O}}}
\newtheorem{lemma}{Lemma}
\newtheorem{corollary}{Corollary}

\title{On the Solvability of Inductive Problems: A Study in Epistemic Topology}
\author{Alexandru Baltag \qquad\qquad  Nina Gierasimczuk \qquad\qquad Sonja Smets 
\institute{Institute for Logic, Language and Computation\\
University of Amsterdam, The Netherlands}
\email{\qquad\qquad A.Baltag@uva.nl \qquad\qquad  Nina.Gierasimczuk@gmail.com  \quad\quad S.J.L.Smets@uva.nl}
}
\def\titlerunning{On the Solvability of Inductive Problems}
\def\authorrunning{A. Baltag, N. Gierasimczuk \& S. Smets}
\begin{document}
\maketitle

\begin{abstract}
We investigate the issues of inductive problem-solving and learning by doxastic agents. We provide topological characterizations of solvability and learnability, and we use them to prove that AGM-style belief revision is ``universal", i.e., that every solvable problem is solvable by AGM conditioning.
\end{abstract}

\section{Introduction}

When in the course of observations it becomes necessary for agents to arrive at a generalization, they should declare, along with their conjecture, the extent of their certainty.
The problem of induction seems formidable if a standard of absolute certainty is imposed on the learner.
Indeed, as is well-known in Philosophy of Science, the so-called problem of empirical underdetermination (i.e., the fact that typically the data are compatible with more than one hypothesis) rules out any chance of obtaining infallible knowledge in empirical research.
But apart from the conclusions based on absolute certainty (cf. \cite{Gie10,DG11,Gierasimczuk:2012aa}), learners can produce hypotheses
based on \emph{beliefs}. It is thus strange that Formal Learning Theory and Belief Revision Theory developed completely independently from each other, and that they have generally maintained their distance ever since.

However, there does exist a line of research that combines belief revision with learning-theoretic notions, line pursued by Kelly \cite{Kel98,Kel98a}, Kelly, Schulte and Hendricks \cite{KSH95}, Martin and Osherson \cite{MO97} and ourselves \cite{Gie10,BGS11,Baltag:2014ac,Gierasimczuk:2013aa}. In this paper we continue this research program, using topological characterizations and methods.

An \emph{inductive problem} consists of a state space, a family of ``potential observations", and
a ``question" (i.e., a partition of the state space). These observations provide data for learning. The problem is \emph{solvable} if there exists a learner that, after observing ``enough" pieces of data, eventually stabilizes on the correct answer.
A special case of solvability is \emph{learnability in the limit}, corresponding to the solvability of the ``ultimate" question: `What is the actual state of the world?'. This notion matches the usual learning-theoretic concept of \emph{identifiability in the limit} \cite{Sol64ab,Gold:1965aa,Put65}.

The aim of the paper is twofold. First, we give topological characterizations of the notions of solvability (and learnability), in terms of topological separation principles. Intuitively, the ability to reliably learn the true answer to a question, is related to the possibility to ``separate" answers by observations. The second goal is to use these topological results to look at the ``solving power" of well-behaved doxastic agents, such as the ones whose beliefs satisfy the usual $KD45$ postulates of doxastic logic, as well as the standard AGM postulates of rational belief-revision \cite{AGM}. We look at
a particularly simple and canonical type of doxastic agent, who forms beliefs by \emph{AGM conditioning}.

Our main result is that AGM conditioning is \emph{universal for problem-solving}, i.e., that every solvable problem can be solved by AGM conditioning.
This means that (contrary to some prior claims), AGM belief-revision postulates are not an obstacle to problem-solving. As a special case,
it follows that AGM conditioning is also ``universal for learning" (every learnable space can be learned by conditioning).\footnote{This special case is a topological translation of one of our previous results \cite{BGS11,Baltag:2014ac}. However, the result about problem-solving universality is not only new and much more general, but also much harder to prove, involving new topological notions and results.}

The close connections between Epistemology and General Topology have already been noticed long ago \cite{Steven-Vickers:1996aa, Kell96}.
Based on these connections, Kevin Kelly started a far-reaching program \cite{Kell96,Kel08} meant to import ideas and techniques from both Formal Learning Theory and Topology into mainstream Epistemology, and show their relevance to the induction problem in Philosophy of Science.
A further connection is the one with Ockham's Razor, that would
\begin{quote}\emph{(...) guarantee that always choosing the simplest theory compatible with experience and hanging on to it while it remains the simplest is both necessary and sufficient for efficiency of inquiry.}
{\cite{Kel08}}\end{quote}
Simplicity has been claimed to have topological characteristics---the simplicity order should in some way follow the structure imposed on the uncertainty range by possible tests and observations. It has also been linked with the notion of minimal mind change, where the learning agent keeps the conjecture changes to a minimum \cite{Kell96,Schulte:1996aa}.

Taken together, our results can be seen as a vindication both of the general topological program in Inductive Epistemology \cite{Kell96,Kel08} and of the AGM Belief Revision Theory \cite{AGM}. On the first front, our general topological characterizations of learning-theoretic concepts seem to confirm Kelly's long-standing claim that Inductive Epistemology can be seen mathematically as a branch of General Topology. On the second front, our universality result seems to vindicate Belief Revision Theory as a canonical form of learning.\footnote{And in the same time (if we adopt a ``simplicity" interpretation of the prior), this last result can be seen as a vindication of Ockham's razor (in line with Kevin Kelly's program).}

\section{Epistemic Spaces and Inductive Problems}\label{epistemic_spaces}

\begin{definition}
An \emph{epistemic space} is a pair $\ES=(S,\Obs)$ consisting of a state space $S$
and a countable (or finite) set of \emph{observable properties} (``data") ${\Obs}\subseteq \mathcal{P}(S)$.
We denote by by $\Obs_s := \{O\in \Obs~|~ s\in O\}$ the set of all observable properties (holding) at a given state $s$.
\end{definition}
One can think of the states in $S$ as ``possible worlds", in the tradition of Kripke and Lewis. The sets $O\in \Obs$ represent properties of the world that are in principle observable: if true, such a property will eventually be observed (although there is no upper bound on the time needed to come to observe it).

To keep things simple, we assume that at each step of the learning process only one property is observed. As for the countability of the set $\Obs$, it is natural to think of observables as properties which can be expressed by means of a language or numerical coding system, generated from a grammar with a finite vocabulary. Any such family $\Obs$ will be (at most) countable.

We denote by ${\mathcal O}^{\cap}$ the family of all finite intersections of observations from ${\mathcal O}$, and by ${\mathcal O}^\ast$ the family of all finite sequences of observations. Such a finite sequence $\sigma=(O_0, O_1,\ldots, O_i)\in {\mathcal O}^\ast$ is called a \emph{data sequence}, and its $i$-th component is denoted by $\sigma_i:=O_i$. It is easy to see that both  ${\mathcal O}^{\cap}$ and ${\mathcal O}^\ast$ are countable.

A \emph{data stream} is a countable sequence $\vec{O}=(O_0, O_1, \ldots)\in \Obs^{\Nat}$ of data from $\Obs$ (here, $\Nat$ is the set of natural numbers, so $\Obs^{\Nat}$ is the set of all maps assigning an observable property to every natural number).
We use the following notation: $\vec{O}_n$ is the $n$-th element in $\vec{O}$; $\vec{O}[n]$ is the initial segment of $\vec{O}$ of length $n$, $(O_0, \ldots, O_{n-1})$; $\Set(\vec{O}):=\{O~|~O\text{ is an element of } \vec{O}\}$ is the set of all data in $\vec{O}$;  $\ast$ is the concatenation operator on strings.

The intuition is that at stage $n$ of a data stream, the agent observes the information in
$O_n$.
A data stream captures a possible future history of
observations in its entirety, while a data sequence captures only a
finite part of such a history.

Given a state $s\in S$, a \emph{data stream for $s$} is a stream  $\vec{O}\in \Obs^{\Nat}$ such that $\Obs_s = \{O\in \Obs~|~ \bigcap_{i=0}^n O_i\subseteq O
\text{ for some } n\in \Nat\}$. Such a stream is ``sound" (every data in $\vec{O}$ is true at $s$) and ``complete" (every true data is entailed by some finite set of observations in $\vec{O}$).

\begin{example}\label{reals}
Let our epistemic space $\ES=(S,\Obs)$ be the real numbers, with observable properties given by open intervals with rational endpoints:
$S:= \Real$, $\Obs:= \{(a,b)~|~ a, b\in Q, a\leq b\}$, where $(a,b):= \{x\in \Real~|~ a<x<b\}$. For instance, observables may represent measurements of a physical quantity (such as a position along a one-dimensional line) that takes real numbers as its possible values.
In such case, for any state $x\in \Real$ and any two sequences $a_n , b_n\in Q$ of rational numbers, such that $a_n\leq x\leq b_n$ and both sequences converge to $x$, the sequence $(a_0, b_0), \ldots, (a_n, b_n), \ldots$ is a (sound and complete) data stream for $x$.

Other examples include \emph{standard $n$-dimensional Euclidean spaces}, e.g., $S=R^3$ with $\Obs$ consisting of all open balls with rational radius and center.
\end{example}

\begin{definition}
An \emph{inductive problem} is a pair $\Prob=(\ES, \mathcal{Q})$ consisting of an epistemic space
$\ES = (S,\Obs)$ together with a ``question" $\mathcal{Q}$, i.e., a partition\footnote{This means that $\bigcup_{i\in I} A_i=S$, and $A_i\cap A_j=\emptyset$ for all $i\not=j$.} of $S$. The cells $A_i$ of the partition $\mathcal{Q}$ are called \emph{answers}. Given $s\in S$, the unique $A\in {\mathcal Q}$ with $s\in A$ is called \emph{the answer to ${\mathcal Q}$ at $s$}, and denoted $A_s$.
We say that a problem $\Prob'=(\ES, \mathcal{Q'})$ is a \emph{refinement} of another problem $\Prob=(\ES, \mathcal{Q})$ (or that the corresponding question ${\mathcal Q'}$ is a \emph{refinement} of the question ${\mathcal Q}$) if every answer of ${\mathcal Q}$ is a disjoint union of answers of ${\mathcal Q'}$.
\end{definition}
The \emph{most refined} question concerns the identity of the real world.
\begin{example}\label{que_id}
The \emph{learning question} on a space $S$
is ${\mathcal Q} =\{\{s\}~|~ s\in S\}$ and corresponds to `What is the actual state?'.
\end{example}

\begin{example}\label{refinement}
Let $\ES=(S, {\mathcal O})$, where $S=\{s,t,u,v)$, ${\mathcal O}=\{U,V, P,Q\}$, with $U=\{s,t\}$, $V=\{s\}$, $P=\{u,v\}$, $Q=\{u\}$.
Take the problem $\Prob= (\ES, {\mathcal Q})$, given by the question ${\mathcal Q}=\{ \{t,u\}, \{s,v\}\}$ depicted on the left-hand side of Figure~\ref{Counterexample}. This can obviously be refined to obtain the problem  $\Prob'= (\ES, {\mathcal Q'})$ given by the learning question
 ${\mathcal Q}=\{ \{s\}, \{t\}, \{u\}, \{v\}\}$ for this space, as depicted on the right-hand side of Figure~\ref{Counterexample}.
\end{example}
\begin{figure}
        \centering
        \def\Qset{(0,0) ellipse (2 and 1) }
\def\Pset{(-1,0) ellipse (3 and 1.5)}
\def\Uset{(-1,-4) ellipse (3 and 1.5)}
\def\Vset{(-2,-4) ellipse (2 and 1)}

\tikzstyle{P} = [draw,black,thick]
                \begin{tikzpicture}[scale=0.5]
\tikzset{
    >=stealth'}

\draw[fill=black,color=black] (-3,0) node[below right] {${t}$} circle (2pt);
\draw[fill=black,color=black] (1,0) node[below right] {${s}$} circle (2pt);
\draw[fill=black,color=black] (-3,-4) node[below right] {${u}$} circle (2pt);
\draw[fill=black,color=black] (1,-4) node[below right] {${v}$} circle (2pt);

    \begin{scope}[rotate=0]
        \path[P] (0,0) \Qset ;

    \end{scope}

        \begin{scope}[rotate=0]
        \path[P] (0,-4) \Uset ;

    \end{scope}
            \begin{scope}[rotate=0]
        \path[P] (-1,-4) \Vset ;

    \end{scope}

     \begin{scope}[rotate=0]
        \path[P] (-1,0) \Pset ;

        \end{scope}

\node[fill=white] (q) at (-1.7,0.5) {${V}$};
\node[fill=white] (p) at (-3.3,0.9) {${U}$};
\node[fill=white] (p) at (1.5,-3) {${P}$};
\node[fill=white] (p) at (-0.3,-3.5) {${Q}$};
\draw[red] (-1,2) -- (-1,-6);

\end{tikzpicture}
\hspace{1.5cm}
                \begin{tikzpicture}[scale=0.5]
\tikzset{
    >=stealth'}

\draw[fill=black,color=black] (-3,0) node[below right] {${t}$} circle (2pt);
\draw[fill=black,color=black] (1,0) node[below right] {${s}$} circle (2pt);
\draw[fill=black,color=black] (-3,-4) node[below right] {${u}$} circle (2pt);
\draw[fill=black,color=black] (1,-4) node[below right] {${v}$} circle (2pt);

    \begin{scope}[rotate=0]
        \path[P] (0,0) \Qset ;

    \end{scope}

        \begin{scope}[rotate=0]
        \path[P] (0,-4) \Uset ;

    \end{scope}
            \begin{scope}[rotate=0]
        \path[P] (-1,-4) \Vset ;

    \end{scope}

     \begin{scope}[rotate=0]
        \path[P] (-1,0) \Pset ;

        \end{scope}

\node[fill=white] (q) at (-1.7,0.5) {${V}$};
\node[fill=white] (p) at (-3.3,0.9) {${U}$};
\node[fill=white] (p) at (1.5,-3) {${P}$};
\node[fill=white] (p) at (-0.3,-3.5) {${Q}$};
\draw[red] (-1,2) -- (-1,-6);
\draw[red] (-4,-2) -- (-1.2,-2);
\draw[red] (-0.8,-2) -- (2,-2);
\end{tikzpicture}
\caption{A problem $\Prob$ (left-hand side) and its refinement $\Prob'$ (right-hand side), see Example~\ref{refinement}}\label{Counterexample}
\end{figure}
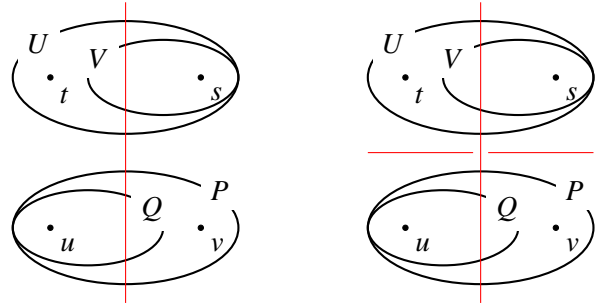

\section{Learning and Problem-Solving}\label{learners_and_conjectures}

\begin{definition}
Let $\ES=(S,\Obs)$ be an epistemic space and let $\sigma_0,\ldots,\sigma_n\in \Obs$. An \emph{agent} (also called a ``learner", or a ``learning method") is a map $\mathcal{L}$ that associates to any
epistemic space $\ES$ and any data sequence $(\sigma_0,\ldots,\sigma_n)$ some family $\mathcal{L}_{\ES} (\sigma_0, \ldots, \sigma_n)\subseteq {\mathcal P} (S)$ of subsets of $S$, satisfying a ``consistency" condition: $\emptyset\not \in  \mathcal{L}_{\ES} (\sigma_0, \ldots, \sigma_n)$ whenever $\bigcap_{i=o}^n \sigma_i\not=\emptyset$.
\end{definition}

Intuitively, after observing the data sequence $\vec{\sigma}=(\sigma_0,\ldots,\sigma_n)$, we can say that agent ${\mathcal L}$ \emph{believes a proposition $P$ after observing the data sequence} $\vec{\sigma}=(\sigma_0,\ldots,\sigma_n)$, and write $B_{\mathcal L}^{\vec{\sigma}} P$ iff
$P\in \mathcal{L}_{\ES} (\sigma_0, \ldots, \sigma_n)$. We can also interpret this as a \emph{conditional belief}, rather than as revised belief, the agent believes every $P\in \mathcal{L}_{\ES} (\sigma_0, \ldots, \sigma_n)$ conditional on $\sigma_0, \ldots, \sigma_n$. However, in the end we are of course interested in the actual revised beliefs after observing the data, so the assumption in this case is that conditional beliefs guide the agent's revision strategy: they ``pre-encode" future belief revisions, to use a term coined by Johan van Benthem \cite{LogicalDynamics}. The above consistency simply means that each of the agent's beliefs is consistent whenever the observed data are consistent.

A \emph{doxastic agent} is one whose set $\mathcal{L}_{\ES} (\sigma_0, \ldots, \sigma_n)$ of beliefs forms a \emph{(proper) filter} on $S$ when observing consistent data; in other words, her beliefs are (consistent when possible, and also) \emph{inference-closed} (i.e., if $P\subseteq Q$ and $P\in \mathcal{L}_{\ES} (\sigma_0, \ldots, \sigma_n)$, then
$Q\in \mathcal{L}_{\ES} (\sigma_0, \ldots, \sigma_n)$) and \emph{conjunctive} (i.e., if $P, Q\in \mathcal{L}_{\ES} (\sigma_0, \ldots, \sigma_n)$ then  $(P\cap Q)\in \mathcal{L}_{\ES} (\sigma_0, \ldots, \sigma_n)$).
Hence, for any doxastic agent ${\mathcal L}$ and every consistent data sequence $\vec{\sigma}$, the belief operator $B_{\mathcal L}^{\vec{\sigma}}$ (as defined above) satisfy the usual $KD45$ axioms of doxastic logic.

A \emph{standard agent} is a doxastic agent $\mathcal{L}$ whose beliefs form a \emph{principal filter}, i.e., all her beliefs are entailed by one ``strongest belief"; formally, a doxastic agent $\mathcal{L}$ is standard
iff for every data sequence $\vec{\sigma}$ over any epistemic space $\ES$ there exists some set $L_{\ES}(\vec{\sigma})$, such that
$$\mathcal{L}_{\ES} (\vec{\sigma}) =\{P\subseteq S~|~ L_{\ES} (\vec{\sigma})\subseteq P\}.$$
It is easy to see that in this case, we must have
$L_{\ES} (\vec{\sigma}) = \bigcap \mathcal{L}_{\ES} (\vec{\sigma})$. Indeed, we can equivalently define a doxastic agent $\mathcal{L}$  to be standard iff $\bigcap \mathcal{L}_{\ES} (\vec{\sigma})\in \mathcal{L}_{\ES} (\vec{\sigma})$ holds for all data sequences $\vec{\sigma}$. \emph{Standard agents
are} \emph{globally consistent} whenever possible: $\bigcap \mathcal{L}_{\ES} (\sigma_0, \ldots, \sigma_n)\not=\emptyset$ whenever $\bigcap_{i=o}^n \sigma_i\not=\emptyset$.

Traditional learning methods in Formal Learning Theory correspond to our standard agents, and they are typically identified with the map $L$ (given by
$L_{\ES} (\sigma_0, \ldots, \sigma_n): =\bigcap \mathcal{L}_{\ES} (\sigma_0, \ldots, \sigma_n)$).
From now on we follow this tradition, and refer to standard agents using the map $L$. But in general we do \emph{not} restrict ourselves to standard agents.

An \emph{AGM agent} is an agent ${\mathcal L}^{\leq}$ who forms beliefs by \emph{AGM conditioning}, i.e., it comes endowed with a map that associates any
epistemic space $\ES$ some total preorder\footnote{A total preorder on $S$ is a binary relation $\leq$ on $S$ that is reflexive, transitive, and connected (i.e., for all $s, t\in S$, we have either $s\leq t$ or $t\leq s$).} $\leq_{\ES}$ on $S$, called \emph{``prior" plausibility relation}; and whose beliefs after observing any data sequence $\vec{\sigma}=(\sigma_0, \ldots, \sigma_n)$ are given by
$${\mathcal L}_{\ES}^{\leq} (\vec{\sigma}):=\{P\subseteq S~|~ \exists s\in \bigcap_{i=0}^n \sigma_i \ \forall t\in \bigcap_{i=0}^n \sigma_i \ (t\leq s\Rightarrow t\in P)\}.$$
Intuitively, $t\leq s$ means that $t$ is \emph{at least as plausible} as $s$ (according to our agent). So, an AGM agent believes $P$
conditional on a data sequence $\vec{\sigma}$ iff $P$ is true in all the states (consistent with the data) that are ``plausible enough".

It is easy to see that \emph{every AGM agent is a doxastic agent}: ${\mathcal L}_{\ES}^{\leq} (\vec{\sigma})$ is a proper filter whenever $\bigcap_{i=0}^n\sigma_i\not=\emptyset$; hence, the beliefs of an AGM agent satisfy the usual $KD45$ axioms of doxastic logic (when learning any consistent data sequence).

Moreover, it is well-known that in fact, \emph{the beliefs of AGM agents satisfy all the so-called AGM axioms from Belief Revision Theory} \cite{AGM}: if, for any
data sequence $\vec{\sigma}=(\sigma_0, \ldots, \sigma_n)$, we set $T={\mathcal L} (\sigma_0, \ldots, \sigma_n)$, and for any new observation $\phi\in {\mathcal O}$ we set $T*\phi = {\mathcal L} (\sigma_0, \ldots, \sigma_n, \phi)$, then the resulting revision operator $*$ satisfies all the AGM postulates. In fact,
for any AGM agent ${\mathcal L}$, if we interpret the operator $B_{\mathcal L}^{\vec{\sigma}}$ (as defined above) as representing a conditional belief $B^{\sigma_0 \wedge\ldots \wedge\sigma_n}$, then the sound and complete logic of these conditional belief operators is the so-called Conditional Doxastic Logic \cite{Boa04,BS08} (which is itself just a repackaging of the AGM postulates in the language of conditional logic).

\begin{observation}
Given a total preorder $\leq$ on $S$ and a subset $A\subseteq S$, set
$$Min_{\leq} (A):= \{s\in A ~|~ s\leq t \mbox{ for all } t\in A\}$$
for the set of $\leq$-minimal states in $A$.
Let $\vec{\sigma}=(\sigma_0, \ldots, \sigma_n)$ be any data sequence such that  $Min_{\leq} (\bigcap_{i=0}^n \sigma_i)\not=\emptyset$.
Then ${\mathcal L}_{\ES}^{\leq} (\vec{\sigma})$ is the principal filter generated by $Min_{\leq} (\bigcap_{i=0}^n \sigma_i)$, i.e., we have
$${\mathcal L}_{\ES}^{\leq} (\sigma_0, \ldots, \sigma_n):= \{P\subseteq S ~|~ Min_{\leq} (\bigcap_{i=0}^n \sigma_i)\subseteq P\}.$$
\end{observation}

In general though, the filter ${\mathcal L}_{\ES}^{\leq} (\vec{\sigma})$ is not principal. So AGM agents are \emph{not} necessarily standard agents. But there is an important case when they are standard:  whenever the preorder $\leq_{\ES}$ is well-founded in every space $\ES$ (i.e., there are no infinite chains $s_0 > s_1 > s_2 \ldots$ of more and more plausible states). It is easy to see that the map $L$ associated to a standard AGM agent is given by the set of $\leq$-minimal states consistent with the data:
$$L_{\ES}^{\leq} (\sigma_0, \ldots, \sigma_n):= Min_{\leq} (\bigcap_{i=0}^n \sigma_i).$$
Intuitively, this means that a standard AGM agent believes a proposition $P$ iff $P$ is true in all the ``most plausible" states consistent with the data.

The original semantics of AGM belief was given using only standard AGM agents. But this semantics was in fact borrowed by Grove \cite{Gro88} from Lewis' semantics for conditionals \cite{Lewis:1969aa}, which did \emph{not} assume well-foundedness.\footnote{Indeed, Lewis' definition of conditionals has a similar shape to our above definition of (conditional) beliefs for non-standard AGM agents.}

\begin{definition} Let $\ES$ be an epistemic space. An agent ${\mathcal L}$ \emph{verifies a proposition $A\subseteq S$ in the limit} if, for every state $s\in S$ and every data stream $\vec{O}$ for $s$, we have $s\in A$ iff there exists some $k\in\Nat$ such that $A\in {\mathcal L}_\ES (\vec{O}[n])$ for all $n\geq k$.
For standard agents, this means that $L_\ES (\vec{O}[n])\subseteq A$ for all $n\geq k$.
A set $A\subseteq S$ is \emph{verifiable in the limit} if there exists some agent that verifies $A$ in the limit.\footnote{For a discussion of the relationship between verifiability and learnability see, e.g., \cite{Kell96,Gie08}.}

An agent ${\mathcal L}$ \emph{falsifies a proposition $A\subseteq S$ in the limit} if, for every state $s\in S$ and every data stream for $\vec{O}$ for $s$, we have $s\notin A$ iff there exists some $k\in\Nat$ such that $A^c\in {\mathcal L}(\ES, \vec{O}[n])\subseteq A^c$ for all $n\geq k$ (here, as in the rest of this paper, $X^c:=S\setminus X$ stands for the complement of $X$). For a standard agent, this means $L(\ES, \vec{O}[n])\subseteq A^c$ for all $n\geq k$,

A proposition $A\subseteq S$ is \emph{falsifiable in the limit} if there exists some agent that falsifies $A$ in the limit.

A proposition $A\subseteq S$ is \emph{decidable in the limit} if it is both verifiable and falsifiable in the limit.

An agent ${\mathcal L}$ \emph{solves a problem} ${\Prob}=(\ES, {\mathcal Q})$ if, for every state $s\in S$ and every data stream $\vec{O}$ for $s$, there exists some $k\in\Nat$ such that $A_s\in {\mathcal L}_{\ES} (\vec{O}[n])$ for all $n\geq k$ (recall that $A_s$ is true answer to ${\mathcal Q}$ at $s$).
For a standard agent, this means that
$L_{\ES}(\vec{O}[n])\subseteq A_s$ for all $n\geq k$. A problem is \emph{solvable (in the limit)} if there exists some agent that solves it.

An epistemic space $\ES=(S, \Obs)$ is \emph{learnable (by an agent ${\mathcal L}$)} if the (problem given by the) learning question ${\mathcal Q}_S=\{ \{s\}~|~ s\in S\}$ is solvable (by ${\mathcal L}$).

All the above notions have a \emph{standard} counterpart, e.g., $A$ is \emph{standardly verifiable} if there exist some standard agent that verifies it; $\Prob$ is \emph{standardly solvable} if it can be solved by some standard agent, etc.
\end{definition}

Note that standard learnability is essentially the same as Gold's \emph{identifiability in the limit} \cite{Put65,Gold67}.

\par\noindent\textbf{Examples and Counterexamples}:
An example of \emph{non-learnable} space $\ES=(S,\Obs)$ is obtained by taking four abstract states $S=\{s,t,u,w\}$ and two observable properties
$\Obs=\{V, U\}$, with $V=\{s,t,u\}$ and $U=\{t,u,w\}$, as depicted in Figure~\ref{example}. Since states $s$ and $t$ satisfy the same observable properties, no learning method will ever distinguish them.
\begin{figure}[htbp]
\begin{center}
\def\Qset{(1,0) ellipse (3 and 1.5) }
\def\Pset{(-1,0) ellipse (3 and 1.5)}

\tikzstyle{P} = [draw,black,thick]
                \begin{tikzpicture}[scale=0.5]
\tikzset{
    >=stealth'}

\draw[fill=black,color=black] (-3,0) node[below right] {${s}$} circle (2pt);
\draw[fill=black,color=black] (-1,0) node[below right] {${t}$} circle (2pt);
\draw[fill=black,color=black] (1,0) node[below right] {${u}$} circle (2pt);
\draw[fill=black,color=black] (3,0) node[below right] {${w}$} circle (2pt);

    \begin{scope}[rotate=0]
        \path[P] (1,0) \Qset ;

    \end{scope}

     \begin{scope}[rotate=0]
        \path[P] (-1,0) \Pset ;

        \end{scope}

\node[fill=white] (q) at (3.3,0.9) {${U}$};
\node[fill=white] (q) at (-3.3,0.9) {${V}$};

\end{tikzpicture}

\end{center}
\caption{A non-learnable space}\label{example}
\end{figure}
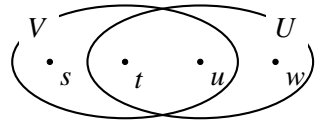

But even spaces in which no two states satisfy the same observations can still be non-learnable, e.g., \emph{all the $n$-dimensional Euclidean spaces from Example~\ref{reals} are} \emph{not learnable} (though, as we will see, many questions are solvable and many subsets are decidable over these spaces).
Another example of \emph{non-learnable} space is given in Figure~\ref{fig_gold_neg}: formally, $\ES=(S,\Obs)$, where
$S:= \{s_n~|~n\in \Nat\}\cup\{s_\infty\}$, and $\Obs=\{O_i~|~i\in\Nat\}$, and for any $i\in\Nat$, $O_i:=\{s_i, s_{i+1},\ldots\}\cup\{s_\infty\}$.

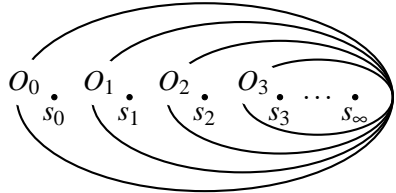
\begin{figure}[h]
\def\Qset{(7,0) ellipse (2 and 1) }
\def\Pset{(6,0) ellipse (3 and 1.5)}
\def\PPset{(5,0) ellipse (4 and 2)}
\def\PPPset{(4,0) ellipse (5 and 2.5)}

\tikzstyle{P} = [draw,black,thick]

\begin{center}
\begin{tikzpicture}[scale=0.5]
\tikzset{
    >=stealth'}

\draw[fill=black,color=black] (0,0) node[below] {${s_0}$} circle (2pt);
\draw[fill=black,color=black] (2,0) node[below] {${s_1}$} circle (2pt);
\draw[fill=black,color=black] (4,0) node[below] {${s_2}$} circle (2pt);
\draw[fill=black,color=black] (6,0) node[below] {${s_3}$} circle (2pt);

\draw[fill=black,color=black] (8,0) node[below] {${s_\infty}$} circle (2pt);

        \path[P] (0,0) \Qset ;
        \path[P] (1,0) \Pset ;
        \path[P] (2,0) \PPset ;
        \path[P] (3,0) \PPPset ;

\node[fill=white, circle, inner sep=1pt] (p1) at (-0.8,0.4) {${O_0}$};
\node[fill=white, circle, inner sep=1pt] (p2) at (1.2,0.4) {${O_1}$};
\node[fill=white, circle, inner sep=1pt] (p3) at (3.2,0.4) {${O_2}$};
\node[fill=white, circle, inner sep=1pt] (p4) at (5.3,0.4) {${O_3}$};
\node[fill=white, circle, inner sep=1pt] (dots) at (7,0) {${\ldots}$};

\end{tikzpicture}

\caption{Another non-learnable space}\label{fig_gold_neg}
\end{center}
\end{figure}
In contrast, an example of \emph{learnable} space is in Figure~\ref{fig_ex_learn_1}:
formally, $S= \{s_{n}~|~ n\in {\omega}\}$ consists of countably many distinct states, with ${\mathcal{O}} =\{
O_n ~|~ n\in {\omega}\}$, where $O_n =\{s_0,s_1,s_2, \ldots, s_n\}$.
\begin{figure}[h]
\def\Qset{(0.02,0) ellipse (1 and 0.5) }
\def\Pset{(1.01,0) ellipse (2 and 1)}
\def\PPset{(2,0) ellipse (3 and 1.5)}
\def\PPPset{(3,0) ellipse (4 and 2)}
\def\PPPPset{(4,0) ellipse (5 and 2.5)}
\tikzstyle{P} = [draw,black,thick]

\begin{center}
\begin{tikzpicture}[scale=0.6]
\tikzset{
    >=stealth'}

\draw[fill=black,color=black] (0,0) node[below] {${s_0}$} circle (2pt);
\draw[fill=black,color=black] (2,0) node[below] {${s_1}$} circle (2pt);
\draw[fill=black,color=black] (4,0) node[below] {${s_2}$} circle (2pt);
\draw[fill=black,color=black] (6,0) node[below] {${s_3}$} circle (2pt);
\draw[fill=black,color=black] (8,0) node[below] {${s_4}$} circle (2pt);

        \path[P] (0,0) \Qset ;
        \path[P] (1,0) \Pset ;
        \path[P] (2,0) \PPset ;
        \path[P] (3,0) \PPPset ;
        \path[P] (4,0) \PPPPset ;

\node[fill=white, circle, inner sep=1pt] (p1) at (0.8,0.4) {${O_0}$};
\node[fill=white, circle, inner sep=1pt] (p2) at (2.8,0.4) {${O_1}$};
\node[fill=white, circle, inner sep=1pt] (p3) at (4.8,0.4) {${O_2}$};
\node[fill=white, circle, inner sep=1pt] (p4) at (6.8,0.4) {${O_3}$};
\node[fill=white, circle, inner sep=1pt] (p5) at (8.8,0.4) {${O_4}$};
\node[fill=white, circle, inner sep=1pt] (dots) at (10.2,0) {${\ldots}$};

\end{tikzpicture}

\caption{A learnable space}\label{fig_ex_learn_1}
\end{center}
\end{figure}
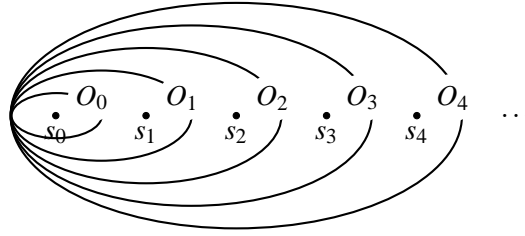
A standard agent that can learn this space in the limit is given by setting $L(\sigma_1, \ldots, \sigma_n)$ to be the \emph{maximum number (in the natural order)} in $\bigcap_{i=0}^n \sigma_i$, whenever there is such a maximum number, and setting $L(\sigma_1, \ldots, \sigma_n):= \bigcap_{i=0}^n \sigma_i$ otherwise.

\begin{proposition}\label{standard}
Let $\ES$ be an epistemic space, $A\subseteq S$ a proposition and $\Prob=(\ES, {\mathcal Q})$ an inductive problem. Then we have the following:
\begin{itemize}
\item $A$ is verifiable (falsifiable, decidable) in the limit iff it is standardly verifiable (falsifiable, decidable) in the limit.
\item $\Prob$ is solvable iff it is standardly solvable.
\item $\ES$ is learnable iff it is standardly learnable.
\end{itemize}
\end{proposition}
\begin{proof}
Let $A\subseteq S$ be a set that is verifiable (falsifiable, decidable) by an agent ${\mathcal L}$ on an epistemic space $\ES$. We construct a standard agent that does the same thing, by setting, for every data sequence $\vec{\sigma}\in {\mathcal O}^*$: $L_\ES (\vec{\sigma}):= A$ if $A\in {\mathcal L} _\ES (\vec{\sigma})$, $L_\ES (\vec{\sigma}):= A^c$ if $A\not\in {\mathcal L} _\ES (\vec{\sigma})$ but $A^c\in {\mathcal L} _\ES (\vec{\sigma})$, and
$L_\ES (\vec{\sigma}):=S$ otherwise. Also, on any \emph{other} space $\ES'=(S', {\mathcal O}')$, we set by default $L_{\ES'} (\vec{\sigma}'):=S'$.

Similarly, let $\Prob=(\ES, {\mathcal Q})$ be a problem that is solvable by ${\mathcal L}$. Let $\leq$ be some arbitrary well-order of the set ${\mathcal Q}$. (Such a well-order exists, by the Well-Ordering Theorem.)
We construct a standard agent who also solves $\Prob$, by setting $L_\ES (\vec{\sigma}):= A$ if $A$ is the first answer in ${\mathcal Q}$ (according to $\leq$) such that $A\in {\mathcal L} _\ES (\vec{\sigma})$ holds; and $L_\ES (\vec{\sigma}):=S$ if no such answer exists. (As before, we can extend our agent to any other space $\ES'=(S', {\mathcal O}'$), by setting $L_{\ES'} (\vec{\sigma}'):=S'$.)

By applying this to the learning problem ${\mathcal Q}= \{\{s\}~|~s\in S\}$, we obtain the similar result for learnability.\end{proof}

In conclusion, everything that can be learned by any agent can also be learned by some standard agent. However, this is no longer true when we restrict to more canonical types of agents (such as AGM agents).
\begin{proposition}\label{non-universal standard}
There exist spaces that are learnable, but not learnable by standard AGM agents. Hence, there exist solvable problems that are not solvable by standard AGM agents.
\end{proposition}
\begin{proof} Consider a counterexample from \cite{Gie10,BGS11,Baltag:2014ac}.
Take the epistemic model from Figure~\ref{fig_ex_learn_1}.
This space is learnable, and thus learnable by $AGM$ conditioning, but it is not learnable by standard conditioning.
Indeed, this space is learnable by conditioning \textit{only} with respect to the following \textit{non-wellfounded} prior:
$s_0>s_1> \ldots > s_n > s_{n+1}>\ldots$ \end{proof}

\section{The Observational Topology}\label{top_notions}

In this section, we assume familiarity with the following notions: \emph{topology} $\tau$ (identified with its family of \emph{open} subsets) over a set $S$ of points, \emph{topological space} $(S, \tau)$, \emph{open} sets, \emph{closed} sets, \emph{interior} $Int(X)$ and \emph{closure} $\closure{X}$ of a set $X$, (open) \emph{neighborhood} of a point $s$, \emph{base} of a topology and \emph{local base (of neighborhoods) at a point}. We use letters $U$, $U'$, etc., for open sets in $\tau$, and letters $C$, $C'$, etc., for closed sets.

A space is said to be \emph{second-countable} if its topology has a countable base.
Given a topological space $(S, \tau)$, the \emph{specialization preorder} ${\sqsubseteq}\subseteq S\times S$ is defined in the following way: for any $s,t\in S$, we set $$\ s\sqsubseteq t\,\,\, \mbox{ iff } \,\,\, \forall U\in\tau \, (s\in U \Rightarrow t\in U).$$

\par\noindent\textbf{Separation Principles}.
In this paper we use four key topological separation notions. The first is the well-known separation axiom $T0$, which will be satisfied by all the topologies that arise in our setting. The second is the separation axiom $TD$. This condition (together with countability) will be shown to characterize \emph{learnable} spaces. The next two notions are analogues of $TD$ separation for \emph{questions}. Instead of asking for open sets that separate points (states), these conditions require the existence of open sets that separate \emph{answers} (to the same question). The concept of \emph{locally closed questions} is a first analogue of $TD$, and it will be shown to characterize in some sense \emph{solvable} problems. Finally, the notion of \emph{linearly separated questions} is a stronger analogue of $TD$ for questions, which characterizes a stronger type of solvability, what we will call \emph{direct solvability by (AGM) conditioning}.

\begin{definition}\label{specialization}
A topological space $(S, \tau)$ satisfies \emph{the separation axiom $T0$} if the specialization preorder is actually a partial order, i.e., it is antisymmetric: $s\sqsubseteq t\sqsubseteq s$ implies $s=t$. Equivalently, if $s\not=t$, then there exists some ``separating" open $U$, such that either $s\in U$, $t\not\in U$, or $s\not\in U$, $t\in U$.

The space $(S, \tau)$ satisfies \emph{the separation axiom $TD$} iff for every point $s\in S$, there is an open $U_x\ni x$ such that $y\not \sqsubseteq x$ for all $y\in O_x\setminus\{x\}$. Equivalently: for every $s\in S$ there is an open $U\in\tau$ such that
$\{s\}=U\cap\closure{\{s\}}$.
\end{definition}
Essentially, $T0$ says that every two points $s\not=t$ can be separated (by an open $U$) one way or another (i.e., either $s\in U$, $t\not\in U$, or $s\not\in U$, $t\in U$), while $TD$ essentially says that every point $s$ can be separated (by an open neighborhood) from all the points $t\not=s$ that are inseparable from $s$.\footnote{A point $y$ is ``inseparable" from $x$ if every open neighborhood of $y$ contains $x$, i.e., $y$ and $x$ are in the topological refinement order $y\sqsubseteq x$.}

\begin{definition}
 Given a topological space $(S,\tau)$, a set $A\subseteq S$ is \emph{locally closed} if it is the intersection $A=U\cap C$ of an open set $U$ with a closed set $C$. Equivalently, if it is of the form $A=U\cap \closure{A}$ for some open $U$.

A set is \emph{$\omega$-constructible} if it is a countable union of locally closed sets.

A question ${\mathcal Q}$ (partition of $S$) is \emph{locally closed} if all its answers are locally closed.
A problem ${\Prob}$ is locally closed if its associated question is locally closed.
\end{definition}
Essentially, locally closed questions are partitions with the property that every ``answer" (i.e., partition cell) $A$ can be separated (by an open neighborhood) from all the non-$A$-states that are inseparable from $A$.\footnote{Here, a state $t$ is said to be ``inseparable" from a set $A$ if there is no open neighborhood $U\ni t$ that is disjoint from $A$.
}

\begin{definition}
A question $\mathcal{Q}$  is \emph{linearly separated} if there exists some \emph{total order} $\unlhd$ on the answers in $\mathcal{Q}$, such that
$A \cap \closure{\bigcup_{B\lhd A} B}=\emptyset$. In other words,
\emph{every answer $A$ can be separated (by some open $U_A\supseteq A$) from the union of all the previous answers}: $U_A\cap B=\emptyset$ for all $B\lhd A$.
\end{definition}
Essentially, a linearly separated question is one whose answers can be totally ordered by a ``plausibility'' (or ``simplicity'') order, in such a way that every answer $A$ can be separated (by an open neighborhood $U_A\supseteq A$) from all answers that are ``more plausible'' (or ``simpler'') than $A$.

\begin{definition}
The \emph{observational topology} $\tau_\ES$ associated with an epistemic space $\ES=(S, \Obs)$ is the topology generated by $\Obs$ (i.e., the smallest collection of subsets of $S$, that includes $\Obs\cup\{\emptyset, S\}$ and is closed under finite intersections and arbitrary unions).
\end{definition}
From now on, we will always implicitly consider our epistemic spaces $\ES$ to also be topological spaces $(S, \tau_{\ES})$, endowed with their observational topology $\tau_{\ES}$. Every topological property possessed by the associated topological space will thus be also attributed to the epistemic space.

\begin{observation}
Every epistemic space is $T0$ and second-countable. A (sound and complete) data stream for $s$ is the same as a local neighborhood base at $s$.
\end{observation}

\begin{proposition}\label{omega_countable_union} Every $\omega$-constructible set can be written as a disjoint countable union of locally closed sets.
\end{proposition}
\begin{proof}
In order to prove this, we first recall some standard topological notions and results: A set is called \emph{constructible} if it is a finite disjoint union of locally closed sets. Obviously, all locally closed sets are constructible. It is known that constructible sets form a Boolean algebra, i.e., the family of constructible sets is closed under complementation, finite unions, and finite intersections.

Suppose $A=\bigcup_{i\in \Nat} A_i$, where all $A_i$ are locally closed. Then we can rewrite $A$ as a disjoint union $A=\bigcup_{i\in \Nat} B_i$, where we have set $B_i=A_i \setminus (\bigcup_{k < i} A_k)= A_i \cap \bigcap_{k<i} A_k^c$, for every $i$. Since $B_i$'s are generated from locally closed sets using complementation and finite intersections, they must be constructible. Hence, each $B_i$ can be written as disjoint finite unions of locally closed sets $B_i =\bigcup_{1\leq j\leq i} B_{ij}$. Hence, we can write $A=\bigcup_{i\in \omega}\bigcup_{1\leq j\leq i} B_{ij}$ as a disjoint countable union of locally closed sets.\end{proof}

\begin{definition}
A \emph{pseudo-stratification}  is a finite or $\omega$-long sequence of locally closed sets $\langle A_i~|~i<\lambda\rangle$ (where $\lambda\in\omega\cup\{\omega\}$), which form a partition of $S$ satisfying the following condition:
$$\mbox{ if $j<i$ then either $A_i\cap \closure{A_j} =\emptyset$ or $A_i\subseteq \closure{A_j}$.}$$
\end{definition}

\begin{proposition}\label{countable_part_into_stratified_part}
Every countable locally closed question can be refined to a pseudo-stratification.
\end{proposition}
\begin{proof} Suppose $\Pi= \{A_i~|~ i\in\Nat\}$ is a countable locally closed question (partition of $S$). We first show the following:

\smallskip

\par\noindent\textbf{Claim.} There exists a family $\{ (\Pi_i, <_i)~|~ i\in \Nat\}$, satisfying
\begin{enumerate}
\item each $\Pi_i$ is a finite partition of $A_i$ into locally closed sets;
\item each $<_i$ is a total order on $\Pi_i$;
\item if $j<i$, $E\in \Pi_j$, $B\in \Pi_i$, then either $B\subseteq \closure{E}$ or $B\subseteq \closure{E}^c$;
\item if $B, E\in \Pi_i$, $E <_i B$, then $B\subseteq \closure{E}^c$.
\end{enumerate}
\emph{Proof of Claim}: We construct $(\Pi_n, <_n)$ by recursion: for $n=0$, set $\Pi_0:=\{A_0\}$, with $<_0$ trivial. For the step $n+1$: assume given $\{(\Pi_i, {<}_i)~|~ i\leq n\}$ satisfying the above four conditions (for $i\leq n$). We set
$$\Pi_{n+1}:=\{ B_f~|~f: \bigcup_{i=1}^n \Pi_i \to \{0,1\}\},$$
where for each function $f: \bigcup_{i=1}^n \Pi_i \to \{0,1\}$ we have set
$$B_f:= A_{n+1} \cap \bigcap \{ \closure{E} ~|~ E\in f^{-1}(0)\}\cap \bigcap \{ \closure{E} ^c ~|~ E\in f^{-1}(1)\}.$$
It is obvious that the $B_f$'s are locally closed (given that $A_{n+1}$ is locally closed) and that they form a partition of $A_{n+1}$. So condition $(1)$ is satisfied.

It is also easy to check condition $(2)$ for $i=n+1$: let $j< n+1$, $E\in \Pi_j$ and $B_f\in \Pi_{n+1}$. Then we have either $f(E)=0$, in which case $B_f\subseteq \closure{E}$ (by construction of $B_f$), or else $f(E)=1$, in which case $B_f\subseteq\closure{E}^c$.

To construct the order $<_{n+1}$, observe first that there is a natural total order $<^{(n)}$ on the disjoint union $\bigcup_{i=1}^n \Pi_i$, namely the one obtained by concatenating the orders $<_0$, $<_1$, $\ldots$ , $<_n$. More precisely, if, for every $B\in \bigcup_{i=1}^n \Pi_i$, we set $i(B)$ to be the unique index $i\leq n$ such that $B\in \Pi_i$, then the order $<^{(n)}$ is given by setting:
$B <^{(n)} E$ iff either $i (B)< i(E)$, or else $i(B)=i(E)$ and $B <_{i(B)} E$.

Now, the order $<_{n+1}$ on $B_f$'s is given by the lexicographic order induced by $<^{(n)}$ on the functions $f$ (thought as ``words" written with the letters $0$ and $1$). More precisely, we set:
$$B_f <_{n+1} B_g$$
iff there exists some set $ E\in \bigcup_{i=1}^n \Pi_i$ such that
$$\left( \forall E'<^{(n)}E \, f(E')=g(E') , \mbox{ but } f(E) < g(E)\right),$$
where $<$ is the usual order $0<1$ on $\{0,1\}$.
Clearly, $<_{n+1}$ is a total order on $\Pi_{n+1}$, so condition $(2)$ is satisfied.

Finally, we check condition $(4)$ for $n+1$, let $B_f, B_g\in \Pi_{n+1}$ such that $B_f <_{n+1} B_g$. By definition of the order $<_{n+1}$, this means that there exists some $E\in \bigcup_{i=1}^n \Pi_i$ such that for all $E'<^{(n)} E$ we have $f(E')=g(E')$ but $f(E)< g(E)$, i.e., $f(E)=0$ and $g(E)=1$. By the construction of $B_f$'s, $f(E)=0$ implies that $B_f\subseteq \closure{E}$, from which we get $\closure{B_f}\subseteq \closure{E}$, and thus $\closure{E}^c\subseteq \closure{B_f}^c$. Similarly, $g(E)=1$ implies that $B_g\subseteq \closure{E}^c$. So we have $B_g\subseteq \closure{E}^c \subseteq \closure{B_f}^c$, and thus by transitivity of inclusion we get $B_g \subseteq \closure{B_f}^c$. This completes the proof of our Claim.

\smallskip

Given now the above Claim, we can prove our Lemma by taking as our refined partition
$$\Pi':= \bigcup_{i\in\omega} \Pi_i.$$
Clearly, $\Pi'$ is a refinement of $\Pi$ consisting of locally closed sets. We now define a well-order $<'$ on $\Pi'$ as the concatenation of all the $\leq_i$'s.\footnote{Once again, one can specify this more precisely by first defining $i:\Pi'\to \omega$ by choosing $i(B)$ to be the unique index $i$ such that $B\in \Pi_i$, and finally defining: $B <' E$ iff either $i (B)< i(E)$, or else $i(B)=i(E)$ and $B <_{i(B)} E$.}
Obviously, $<'$ is a total order of type $\leq \omega$ on $\Pi'$, so we get finite or $\omega$-long sequence that enumerates $\Pi'$.
The above properties $(3)$ and $(4)$ ensure that this is a pseudo-stratification.\end{proof}

\begin{lemma} \label{pseudo-strat_map} Given a pseudo-stratification $\langle A_i~|~i<\lambda\rangle$ (of length $\lambda\leq \omega$), there exists a $\lambda$-long sequence of open sets $\langle U_i~|~i <\lambda\rangle$, satisfying:
\begin{enumerate}
\item $U_i \cap \closure{A_i}=A_i$;
\item if $j< i$ and $U_i\cap A_j\not=\emptyset$, then $A_i\subseteq \closure{A_j}$.
\end{enumerate}
\end{lemma}
\begin{proof} We know that each $A_i$ is locally closed, so there exists some open set $U^{A_i} \in\tau$ such that $U^{A_i} \cap \closure{A_i}=A_i$.
Now, for all $i\in\omega$ set
$$U_i:= U^{A_i} \cap \bigcap\{ \closure{A_j}^c ~|~ j < i, A_i \subseteq \closure{A_j}^c \}.$$
Let us first check that the sequence $\langle U_i~|~i < \lambda\rangle$ satisfies condition $(1)$:
$$U_i\cap \closure{A_i}=  (U^{A_i} \cap \bigcap\{ \closure{A_j}^c ~|~ j < i, A_i \subseteq \closure{A_j}^c \} )\cap  \closure{A_i}$$
$$= (U_i \cap \closure{A_i})\cap \bigcap\{ \closure{A_j}^c ~|~ j < i, A_i \subseteq \closure{A_j}^c \}$$
$$=A_i \cap \bigcap\{ \closure{A_j}^c ~|~ j < i, A_i \subseteq \closure{A_j}^c \}=A_i$$
Second, let us check condition $(2)$: Suppose that we have $j<i$ and $U_i\cap A_j\not=\emptyset$, but $A_i\not\subseteq \closure{A_j}$. Since $(A_i)_{i<\lambda}$ is a pseudo-stratified sequence, from $j<i$ and $A_i\not\subseteq \closure{A_j}$ we can derive $A_i\subseteq \closure{A_j}^c$. By the construction of $U_i$, this implies that $U_i\subseteq \closure{A_j}^c$, and hence that $U_i\cap A_j\subseteq \closure{A_j}^c\cap A_j\subseteq  \closure{A_j}^c\cap \closure{A_j}=\emptyset$, which contradicts the assumption that $U_i\cap A_j\not=\emptyset$.\end{proof}

\begin{lemma}
\label{pseudo-strat_linear} Every pseudo-stratification is linearly separated.  \end{lemma}
\begin{proof} Let $\Pi=\{A_i~|~i <\lambda\}$  be a pseudo-stratification (with $\lambda\leq \omega)$, and let $\langle U_i~|~i< \lambda\rangle$  be a sequence satisfying the conditions of Lemma \ref{pseudo-strat_map}. It is clear that, in order to prove our intended result, it is enough to construct a total order $\unlhd$ on the set $\{i\in\omega|i<\lambda\}=\lambda\subseteq\omega$, such that
$$U_i \cap A_j \not=\emptyset \, \Rightarrow\, i \unlhd j.$$

For this, we
first define a reflexive relation $R$ on $\lambda$, by setting
$$i R j \Longleftrightarrow  U_i\cap A_j \not=\emptyset.$$

\emph{Claim}: There are no non-trivial cycles
$$i_1 R \cdots i_n R i_1 \,\, \mbox{ (with distinct $i_k$'s)}.$$

\smallskip

\emph{Proof of Claim}: Let $i_1 R \cdots i_n R i_1$ be a non-trivial cycle of \emph{minimal length} $n\geq 2$. There are two cases:

\smallskip

\textbf{Case 1}: $n=2$, i.e., $i_1 R i_2 R i_1$ with $i_2\not= i_1$. We must have either $i_1 < i_2$ or $i_2 < i_1$. Without loss of generality, we can assume $i_1 < i_2$ (otherwise, just swap $i_1$ and $i_2$, and use the cycle $i_2 R i_1 R i_2$). From $i_2 R i_1$, we get $U_{i_2} \cap A_{i_1}\not=\emptyset$. This together with $i_1 < i_2$, gives us $A_{i_2} \subseteq \closure{ A_{i_1}}$ (by condition $(2)$ from Lemma 2), and hence $U_{i_1} \cap A_{i_2}\subseteq U_{i_1}\cap \closure{A_{i_1}}=A_{i_1}$. From this, we get that $U_{i_1} \cap A_{i_2}= (U_{i_1} \cap A_{i_2}) \cap A_{i_2}\subseteq A_{i_1}\cap A_{i_2}=\emptyset$ (since $i_1\not= i_2$, so $A_{i_1}$ and $A_{i_2}$ are different answers, hence disjoint), so we conclude that  $U_{i_1} \cap A_{i_2}=\emptyset$.
But on the other hand, from $i_1 R i_2$ we get $U_{i_1} \cap A_{i_2}\not=\emptyset$. Contradiction.

\smallskip

\textbf{Case 2}: $n>2$. Since all the $i_k$'s are distinct, there must exist a (unique) smallest index in the cycle.
Without loss of generality (since otherwise we can rearrange the indices, permuting the cycle), we can assume that $i_3$ is the smallest index. (Note that, since $n>2$, there must be at least three distinct successive indices $i_1, i_2, i_3$.) So $i_3<i_1$ and $i_3< i_2$. From $i_2 R i_3$ we get $U_{i_2}\cap A_{i_2}\not=\emptyset$. Since $i_3 < i_2$, it follows that $A_{i_2}\subseteq \closure{A_{i_3}}$ (by Lemma 2). But on the other hand, $i_1 R i_2$ gives us $U_{i_1}\cap A_{i_2}\not=\emptyset$. We hence obtain $U_{i_1}\cap \closure{A_{i_3}}\not=\emptyset$. This, together with $i_3< i_1$, gives us $A_{i_1}\subseteq\closure{A_{i_3}}$ (again by Lemma 2). From this, we derive  $A_{i_1}\subseteq U_{i_1}\cap \closure{A_{i_3}}$ (since $A_i\subseteq U_i$ for all $i$). Let now $s\in A_{i_1}$ be any state satisfying the answer $A_{i_1}\subseteq  U_{i_1}\cap \closure{A_{i_3}}$. So we have $s\in U_{i_1}$ and $s\in \closure{A_{i_3}}$, which together imply that $U_{i_1}\cap A_{i_3}\not=\emptyset$ (since $s\in \closure{A_{i_3}}$ implies that every open neighborhood of $s$ intersects $A_{i_3}$). Hence, we have $i_1 R i_3$, which means we can shorten the cycle by eliminating $i_2$, we obtain contradiction.

\medskip

Given the above Claim, it follows that the transitive closure $R^*$ is a partial order on $\lambda$ (which obviously includes $R$). By the Order Extension Principle, we can extend $R^*$ to a total order $\unlhd$ on $\lambda$, which still includes $R$.\end{proof}

\section{Topological Characterization of Solvability}\label{topological_characterizations}

\begin{definition} Let $\ES=(S,\Obs)$ be an epistemic space, $L$ be a standard agent, $A\subseteq S$, and $s\in A$. An \emph{$A$-locking sequence for $s$ (with respect to $L$)} is a data sequence $\sigma=(O_1, \ldots, O_k)$, such that:
\begin{enumerate}
\item $\sigma$ is sound for $s$, i.e., $s\in \bigcap_{1\leq i\leq k} O_i$;
\item if $\delta$ is any data sequence sound for $s$, then $L(\ES,\sigma\ast\delta)\subseteq A$.
\end{enumerate}

For a given data sequence $\sigma$, we denote by $L_A^{\sigma}$ the set of all states in $A$ having $\sigma$ as an $A$-locking sequence, i.e.,
$$L_A^{\sigma}:=\{s\in A~|~ \sigma \mbox{ is an $A$-locking sequence for $s$ wrt $L$} \}.$$
\end{definition}

\begin{lemma}\label{ver_bigsum} If $A$ is verifiable in the limit by a standard agent $L$, then $\bigcup_{\sigma\in {\mathcal O}^\ast} L_A^{\sigma}=A.$
\end{lemma}
\begin{proof} Suppose not. Let $A$ be verifiable in the limit, but such that $A\not=\bigcup_{\sigma\in{\mathcal O}^*} L_A^{\sigma}$. Since all $L_A^{\sigma}\subseteq A$, his means that $A\not\subseteq\bigcup_{\sigma\in{\mathcal O}^*} L_A^{\sigma}$, i.e., there exists some state
$s\in A$ for which there is no $A$-locking sequence. This means that every data sequence $\sigma$ that is sound for $s$ can be extended to a sequence $\delta$ that is also sound for $s$ and has $L(\delta)\not\subseteq A$.

Let now $\vec{O}$ be a (sound and complete) data stream for $s$. We construct a new infinite data stream $\vec{V}$, by defining increasingly longer initial segments $\delta_k$ of $\vec{O}$, in countably many stages: we first set $V_0=O_0$, thus obtaining an initial segment $\delta_0 =(O_0)=(V_0)$; at the $k+1$-th stage, given some initial segment $\delta_k=(V_0, V_1, \ldots , V_{n_k})$ (of some length $n_k$), we built our next initial segment by taking any extension $\delta_{k+1}$ of the sequence $\sigma_k = (V_0, \ldots, V_{n_k}, O_{n+1})$ that is sound for $s$ and has $L(\sigma_k)\not\subseteq A$. The resulting infinite stream  $\vec{V}$ is a (sound and complete) stream for $s$ (the completeness of $\vec{V}$ with respect to $s$ follows the fact that this stream includes all the elements of $\vec{U}$), but which contains arbitrarily long initial segments $\sigma_k$ with $L(\sigma_k)\not\subseteq A$. Since $s\in A$, this contradicts the assumption that $A$ is verifiable in the limit.\end{proof}

\begin{lemma}\label{ver_loc_closed} If $A\subseteq S$ is verifiable in the limit by a standard agent $L$, then for every data sequence $\sigma=(O_1, \ldots O_k)$, the set $L_A^{\sigma}$ is locally closed.
\end{lemma}
\begin{proof} Let $O:=\bigcap_{i=1}^k O_i$ be the intersection of all the observations in $\sigma$. We will show that $$O \cap \closure{L_A^{\sigma}}= L_A^{\sigma},$$
from which the desired conclusion follows.

($\supseteq$) If $s\in L_A^{\sigma}$, then $\sigma$ is an $A$-locking sequence for $s$, hence $\sigma$ is sound for $s$, and thus $s\in \bigcap_{i=1}^{n} O_i=O$.

($\subseteq$) Suppose that $s\in O\cap \closure{L_A^{\sigma}}$. We prove two claims:

\smallskip

\par\noindent\textbf{Claim 1}: For every data sequence $\delta$ that is sound for $s$ and extends $\sigma$, we have $L_\ES (\delta)\subseteq A$.

\emph{Proof of Claim 1}: Let $\delta= (\delta_1,\ldots, \delta_n)$ be a data sequence that is sound for $s$ (i.e., $s\in \delta_i$ for all $i=1,\ldots,n$) and extends $\sigma$, i.e., $n\geq k$ and $U_i=O_i$ for all $i\leq k$). Hence, $\bigcap_{i=1}^n \delta_i$ is an open neighborhood of $s$, and $s\in \closure{L_A^{\sigma}}$, so there must exist some $t\in \bigcap_{i=1}^n \delta_i$ such that $t\in L_A^{\sigma}$. Hence, $t\in A$ and $\sigma$ is an $A$-locking sequence for $t$. But $\delta$ extends $\sigma$ and is sound for $t$, so (by the definition of $\sigma$ being an $A$-locking sequence for $t$), we have that $L (\delta)\subseteq A$, which concludes the proof of Claim 1.

\smallskip

\par\noindent\textbf{Claim 2}: We have $s\in A$.

\emph{Proof of Claim 2}: Let $\vec{V}$ be a stream for $s$ that extends $\sigma$ (such a stream must exist, since $\sigma$ is sound for $s$: just take any stream for $s$ and prefix it with $\sigma$). Then, for every $n\geq k$, the sequence $\delta_n = (V_1, \ldots, V_n)$ is sound for $s$ and extends $\sigma$. Hence, by the above Claim, we must have that $L_{\ES} (V_1, \ldots, V_n)\subseteq A$ for all $n\geq k$. But we assumed that $A$ is verifiable in the limit, so we must have $s\in A$, which concludes the proof of Claim 2.

\medskip

From Claims 1 and 2 together, we conclude that $\sigma$ is an $A$-locking sequence for $s\in A$, hence $s\in L_A^\sigma$.\end{proof}

\begin{theorem} Given an epistemic space $(S, {\mathcal O})$, a set $A\subseteq S$ is verifiable in the limit iff it is $\omega$-constructible.
\end{theorem}
\begin{proof}
($\Leftarrow$) Assume $A= \bigcup_n (U_n\cap C_n)$ is a countable disjoint union of (mutually disjoint) locally closed sets $U_n\cap C_n$ (with $U_n$ open and $C_n$ closed). We define a standard agent $L$ for $A$ on finite data sequences $\delta= (O_1, \ldots, O_k)$, by setting $L(\ES,\delta)=A^c$, if we have $\bigcap_j O_j \not\subseteq U_n$ for all $n\in\omega$; $L(\ES,\delta)=A^c$ (where $A^c$ is the complement of $A$), if $\bigcap_j O_j \subseteq C_n^c$ holds for the first index $n\in\omega$ such that $\bigcap_j O_j \subseteq U_n$; and $L(\ES,\delta)=A$ otherwise. Then it is easy to see that $L$ verifies $A$ in the limit.

($\Rightarrow$) Suppose that $A$ is verifiable in the limit. By Proposition \ref{standard}, it is then verifiable by a standard agent $L$. By Lemma 1, $A$ is the union of all sets $L_A^{\sigma}$ for all finite data sequences $\sigma$. But there are only countably many such sequences, so this is a countable union. Moreover, by Lemma 2, each $L_A^{\sigma}$ is locally closed. Hence $A$ is a countable union of locally closed sets, i.e., an $\omega$-constructible set.
\end{proof}

\begin{corollary}\label{decidable}
$A$ is decidable in the limit iff both $A$ and $A^c$ are $\omega$-constructible.
\end{corollary}
\begin{proof} Follows trivially from the above results. \end{proof}

\begin{theorem}\label{solvable_omega_refinement} Let ${\Prob}= (\ES,{\mathcal Q})$ be an inductive problem on an epistemic space $\ES$. The following are equivalent:
\begin{enumerate}
\item ${\Prob}$ is solvable (in the limit);
\item the associated question ${\mathcal Q}$ is an (at most) countable family of $\omega$-constructible answers;
\item ${\mathcal Q}$ has an (at most) countable locally closed refinement.
\end{enumerate}
\end{theorem}
\begin{proof}
\par\noindent \textbf{$(1)\Rightarrow (2):$}
Let $\Prob$ be a solvable problem.
By Proposition \ref{standard}, there exists some standard agent that solves it. Let $L$ be such a standard agent that solves $\Prob$.

\par\noindent\textbf{Claim:} Every answer $A\in {\mathcal Q}$ is verifiable in the limit.

\emph{Proof of Claim}: Let $A\in {\mathcal Q}$ be an answer. We construct a standard agent $L^A$ that verifies it, by setting
$L^A_{\ES} (\sigma):= A$ iff $L_{\ES} (\sigma)\subseteq A$, and $L^A_{\ES} (\sigma):= A^c$ otherwise. It is easy to see that $L^A$ verifies $A$.

\medskip

Using the Claim and Lemma \ref{ver_bigsum}, we obtain that, for each answer $A\in {\mathcal Q}$, there exists some data sequence $\sigma\in {\mathcal O}^*$ such that
$L_{\ES} (\sigma)\subseteq A$. But ${\mathcal O}^*$ is countable, so there can be only countably many answers in ${\mathcal Q}$.

By the claim above, Lemma \ref{ver_bigsum} and Lemma \ref{ver_loc_closed}, we obtain that every answer $A\in {\mathcal Q}$ is a countable union of locally closed sets, hence it is $\omega$-constructible.

\smallskip

\par\noindent\textbf{ $(2)\Rightarrow (3):$}
By $(2)$, ${\mathcal Q}$ is (at most) countable, say ${\mathcal Q}=\{A_i~|~i\in \omega\}$, and also each answer $A_I\in {\mathcal Q}$ is $\omega$-constructible, hence it can be written as a countable disjoint union of locally closets $A=\bigcup_{k\in\omega} A_i^k$ (where all $A_i^k$'s locally closed and mutually disjoint). Then the question $\{A_i^k ~|~ i\in \omega, k\in\omega\}$ is a refinement of ${\mathcal Q}$, which is countable and locally closed.

\smallskip

\par\noindent\textbf{ $(3)\Rightarrow (1):$}
Let ${\mathcal Q}'=\{B_i~|~i\in\omega\}$ be a countable closed refinement of ${\mathcal Q}'$. By Corollary \ref{decidable}, every answer $B\in{\mathcal Q}'$ is decidable, and so by Proposition \ref{standard}, we can choose for each $B_i\in{\mathcal Q}$ some standard agent $L_i$ that decides $B_i$. We define now a new standard agent $L$, by:
$$L_{\ES} (\sigma):= \bigcup\{ B_i ~|~i\in \omega \mbox{ such that } L_i (\sigma)\subseteq B_i\}.$$
It is easy to see that this agent $L$ solves ${\mathcal Q}'$, and since ${\mathcal Q}'$ is a refinement of ${\mathcal Q}$, $L$ also solves ${\mathcal Q}$.\end{proof}

\begin{corollary}\label{TD learnability}
An epistemic space $\ES=(S, \Obs)$ is learnable in the limit iff it is countable and satisfies the $TD$ separation axiom.
\end{corollary}
\begin{proof}
Apply Theorem \ref{solvable_omega_refinement} to the learning question $\{\{s\}~|~s\in S\}$, noticing that the fact that all its answers are $\omega$-constructible is equivalent to  all singletons being locally closed, which is just another formulation of the $TD$ axiom.\end{proof}

\section{Universality of Conditioning}

Our aim in this section is to show that $AGM$ conditioning is ``universal": every solvable problem can be solved by some $AGM$ agent.
First, we introduce an auxiliary notion, that of a problem being \emph{directly solvable by AGM conditioning}.

Given a question ${\mathcal Q}$ on an epistemic space $(S, {\mathcal O})$, any total order $\unlhd\subseteq {\mathcal Q}\times {\mathcal Q}$ on (the answers of) the question ${\mathcal Q}$ induces in a canonical way a total preorder $\leq\subseteq S\times S$, obtained by:
$$s \leq t \,\, \mbox{ iff } \,\, A_s \unlhd A_t$$
(where $A_s$ is the unique answer $A_s\in {\mathcal Q}$ such that $s\in A_s$).

\begin{definition}
A problem $\Prob= (\ES, {\mathcal Q})$ is \emph{directly solvable by conditioning} if it is solvable by AGM conditioning with respect to (a prior $\leq$ that is canonically induced, as explained above, by) a total order $\unlhd\subseteq {\mathcal Q}\times {\mathcal Q}$ on (the answers of) the question ${\mathcal Q}$.
\end{definition}

Direct solvability by conditioning essentially means that the problem can be solved by a conditioning agent \emph{who does not attempt to refine the original question}: she forms beliefs only about the answers to the given question, and is thus indifferent between states satisfying the same answer.
Direct solvability by conditioning is thus a very stringent condition, and unsurprisingly this form of conditioning is \emph{not} universal.

\begin{proposition}\label{genin} (\emph{K. Genin, personal communication})
Not every solvable problem is directly solvable by conditioning.
\end{proposition}
\begin{proof}
\par\noindent Let $\Prob$ be the problem in Example~\ref{refinement}, depicted on the left-hand side of Figure~\ref{Counterexample}. It is easy to see that this problem cannot be directly solvable by conditioning. Indeed, if $\{t,u\}\lhd\{s,v\}$ then $v$ is not learnable by $\lhd$-conditioning; if $\{s,v\}<\{t,u\}$ then $t$ is not learnable by $\lhd$-conditioning; while if $\{t,u\}$ and $\{s,v\}$ are equally plausible, then neither $t$ nor $v$ are learnable.

But $\Prob$ \emph{can} be refined to a directly solvable problem, namely the ``learning question" $\Prob'$ (depicted on the right-hand side of Figure~\ref{Counterexample}), which \emph{can} be directly solvable (e.g. if we set $\{t\}\lhd \{s\}\lhd\{v\}\lhd \{u\}$).
As a consequence, $\Prob$ can itself be solved by (non-direct) conditioning (with respect to the order $t<s<v<u$).
\end{proof}

This counterexample suggests a way to prove our intended universality result: it is enough to show that every solvable problem has a refinement that is directly solvable by conditioning. To do this, we first need a structural characterization of direct solvability.

\begin{lemma}\emph{(Topological Characterization of Direct Solvability by Conditioning)}
A problem  $\Prob= (\ES, {\mathcal Q})$ is directly solvable by conditioning iff  ${\mathcal Q}$ is linearly separated.
\end{lemma}

\begin{proof}
\emph{Left-to-right implication}: Suppose that $\Prob$ is directly solvable by conditioning with respect to (a prior $\leq$ that is canonically induced
by) a total order $\unlhd\subseteq {\mathcal Q}\times {\mathcal Q}$. Then, for every $s\in S$ choose some sound and complete data stream $\vec{O}^S=(O^s_n)_{n\in\omega}$ for $s$ (with $O^n_s\in\Obs\subseteq \tau_{\ES}$). Direct solvability by conditioning implies then that there exists some $N_s$ such that $Min_{\leq} (O^s_1, \ldots, O^s_{N_s})\subseteq A_s$. Set $U_s:= \bigcap_{i=1}^{N_s} O^s_i\in \tau_{\ES}$, so that we have $s\in U_s$ and $Min_{\leq} U_s\subseteq A_s$.
Then set $U_A:=\bigcup_{s\in A} U_s\in \tau_{\ES}$ for every answer $A\in {\mathcal Q}$. We claim that $U_A$ ``separates" $A$ from the union of all the answers $B\lhd A$ (as linear separation demands): indeed, by the construction of $U_A$, it is obvious that (1) $A\subseteq U_A$, and also that $Min_{\leq} U_A\subseteq A$. By unfolding the last clause in terms of  $\unlhd$, we obtain that: $A\unlhd B$ holds for all $B\in {\mathcal Q}$ such that $U_A\cap B\not=\emptyset$. Since $\unlhd$ is a total order on ${\mathcal Q}$, this is equivalent to: (2) $U_A\cap B=\emptyset$ for all $B\lhd A$. By (1) and (2) together, we obtain that ${\mathcal Q}$ is linearly separated.

\emph{Right-to-left implication}: Suppose  ${\mathcal Q}$ is linearly separated.  Let $\unlhd$ be a total order on ${\mathcal Q}$ that linearly separates it. This means that, for every answer $A\in {\mathcal Q}$, there exists some open set $U_A\in\tau_{\ES}$ such that $A\subseteq U_A$ and $U_A\cap B=\emptyset$ for all $B\lhd A$.
For each $s\in S$, we set $U_s:=U_{A_s}$ (where $A_s$ is the unique answer $A_s\in {\mathcal Q}$ with $s\in A_s$).

Let $\leq$ be the total preorder on $S$ canonically induced by the order $\unlhd\subseteq {\mathcal Q}\times {\mathcal Q}$ (by
$s \leq t$ iff $A_s \unlhd A_t$).
We show now that \emph{${\Prob}$ is directly solvable by conditioning with respect to $\leq$}.
For this, let $s\in S$ be any state, and $\vec{O}=(O_n)_{n\in\omega}$ be a sound and complete stream for $s$. Completeness of the stream implies that there must exist some $N\in \omega$ such that $\bigcap_{i=1}^N O_i\subseteq U_s$.

To conclude our proof, it is enough to show the following

\medskip

\emph{Claim}: For every $n\geq N$, we have
$$s\in  Min_{\leq} (\bigcap_{i=1}^n O_i)\subseteq A_s.$$

First, let us see why this Claim is enough to give us direct solvability by conditioning. The fact that $s\in  Min_{\leq} (\bigcap_{i=1}^n O_i)$ implies that
$Min_{\leq} (\bigcap_{i=1}^n O_i)\not=\emptyset$, for all $n\geq N$.
A previous observation tells us that, when applied to such data streams, the AGM agent ${\mathcal L}^{\leq}$ produces a ``principal filter", given by
$${\mathcal L}^{\leq} (O_1, \ldots, O_n)=\{P\subseteq S ~|~ Min_{\leq} (\bigcap_{i=1}^n O_i)\subseteq P\}.$$
By the Claim above we have $Min_{\leq} (\bigcap_{i=1}^n O_i)\subseteq A_s$, and hence we obtain $A_s \in {\mathcal L}^{\leq} (O_1, \ldots, O_n)$, for all $n\geq N$.

\medskip

\emph{Proof of Claim}: Let $n\geq N$.  To prove the Claim, it is enough to show the following two implications (for all states $t$):
\begin{enumerate}
\item $t\in \bigcap_{i=1}^n O_i \, \Rightarrow \,s\leq t $;
\item $t\in  Min_{\leq} (\bigcap_{i=1}^n O_i) \, \Rightarrow \, A_t=A_s$.
\end{enumerate}

To show $(1)$, let $t\in \bigcap_{i=1}^n O_i$. Then $t\in U_s$ (since $\bigcap_{i=1}^n O_i\subseteq \bigcap_{i=1}^N O_i\subseteq U_s$),
so $U_s\cap A_t\not=\emptyset$. Hence (by linear separation) we must have $A_s \unlhd A_t$, i.e., $s\leq t$.

To show $(2)$, let $t\in  Min_{\leq} (\bigcap_{i=1}^n O_i)$. This implies that $t\leq s$ (since $s\in \bigcap_{i=1}^n O_i$). But by $(1)$, we also have
$s\leq t$, and hence $s\leq t\leq s$. This means that $A_s \unlhd A_t\unlhd A_s$. But $\unlhd$ is a total order on ${\mathcal Q}$, so it follows that $A_t=A_s$.
\end{proof}

\begin{theorem}\label{question_agm}
AGM conditioning is a \emph{universal problem-solving method}, i.e., every solvable problem is solvable by some AGM agent.
\end{theorem}

\begin{proof}
Let $\Prob$ be a solvable problem. From Theorem 2, Proposition 3 and Lemma 2,
it follows that $\Prob$ has a linearly separated refinement $\Prob'$. By Lemma 5, that refinement is (directly) solvable by an AGM agent $L^{\leq}$. It is obvious (from the definition of solvability) that any doxastic agent which solves the more refined problem $\Prob'$ solves also the original problem $\Prob$.
\end{proof}

\begin{corollary}\label{learning agm}
AGM conditioning is a \emph{universal learning method}, i.e.,
\emph{every learnable space is learnable by some AGM agent}.
\end{corollary}

\begin{proof}
Apply the previous result to the finest question ${\mathcal Q}:=\{\{s\}~|~s\in S\}$.
\end{proof}

In contrast, recall that the counterexample in Proposition~\ref{non-universal standard} showed that
\emph{standard} AGM agents have a very limited problem-solving power. Standard conditioning is not a universal learning method (while general $AGM$ conditioning is universal).
This means that \emph{allowing prior plausibility orders that are non-wellfounded is essential for achieving universality of conditioning}. Beliefs generated in this way may occasionally fail to be globally consistent. (Indeed, note that in the counterexample from Proposition~\ref{non-universal standard}, the beliefs of the non-standard $AGM$ agent who learns the space are initially globally inconsistent. In conclusion, \emph{occasional global inconsistencies are the unavoidable price for the universality of $AGM$ conditioning}.

\section{Conclusions and Connections to Other Work}

The general topological setting for problem-solving assumed here is a variation of the one championed by Kelly in various talks \cite{Kell11} and in unpublished work \cite{Kel15, KellHantiSimplicity}, though until recently we did not realize this close similarity.
Our topological characterizations of verifiable, falsifiable and decidable properties are generalizations of  results by Kelly \cite{Kell96}, who proved characterizations for the special case of Baire spaces.\footnote{In unpublished work \cite{KellHantiSimplicity} the authors claim a characterization of solvability in a general setting. Their characterization is sightly ``looser" than ours, and can be easily obtained from ours. Our tighter characterization is the one needed for proving universality.} Our result on learning-universality (Corollary \ref{learning agm}) is also a generalization of analogue results by Kelly \cite{Kel98,Kel98a}, and Kelly, Schulte and Hendricks \cite{KSH95}. But our generalization to arbitrary spaces is highly non-trivial, requiring the use of the $TD$ characterization.
In contrast, the Baire space satisfies the much stronger separation axiom $T1$, which trivializes the specialization order, and so the proof of learning-universality is much easier in this special case: \emph{any} total $\omega$-like ordering of the space can be used for conditioning.
Nevertheless, in a sense, this result is just a topological re-packaging of one of our own previous results \cite{Gie10, BGS11,Baltag:2014ac}.

While writing this paper, we learned that our $TD$ characterization of learnability (Corollary \ref{TD learnability}) was independently re-proven by Konstantin Genin (\cite{Genin14}, unpublished manuscript), soon after we announced its proof. This characterization is actually a topological translation of a classical characterization of identifiability in the limit \cite{Ang80}, and in fact it also follows from a result by de Brecht and Yamamoto \cite{Yamamoto:2010}, who prove it for so-called ``concept spaces''.

Our key new results are far-reaching and highly non-trivial: the topological characterization of solvability (Theorem \ref{solvable_omega_refinement}), and the universality of AGM condition for problem-solving (Theorem \ref{question_agm}). They required the introduction of new topological concepts (e.g., pseudo-stratifications and linearly separated partitions), and some non-trivial proofs of new topological results.

Philosophically, the importance of these results is that, on the one hand they fully vindicate the general topological program in Inductive Epistemology started by Kelly and others \cite{Kell96,Schulte:1996aa}, and on the other hand they reassert the power and applicability of the AGM Belief Revision Theory against its critics. To this conclusion, we need to add an important proviso: our results show that, in order to achieve problem-solving universality, AGM agents need to (a) be ``creative", by \emph{going beyond the original problem} (i.e., finding a more refined problem that can be solved directly, and forming prior beliefs about the answer to this more refined question), and (b) admit \emph{non-standard priors}, which occasionally will lead to \emph{beliefs that are globally inconsistent} (although still locally consistent). Such occasional global inconsistencies can give rise to a type of ``infinite Lottery Paradox". But this is
the price that AGM agents \emph{have to pay} in order to be able to solve every solvable question. Whether or not this is a price that is worth paying is a different, more vague and more ``ideological" question, although a very interesting one. But this question lies beyond the scope of this paper.

\section{Acknowledgments}

We thank Johan van Benthem, Nick Bezhanishvili, Konstantin Genin, Thomas Icard and Kevin Kelly for their useful feedback on issues related to this paper. Johan helped us place belief-based learning within the larger context of long-term doxastic protocols \cite{MergingFrameworks}, and beyond this he gave us his continuous support and encouragement for our work on this line of inquiry.  Nick pointed to us the connections between our work and the notions of TD-space and locally closed set. Konstantin pointed to us the connections to the notion of stratification and gave the counterexample proving Proposition~\ref{genin}. His critical feedback on our early drafts was really essential for clarifying our thoughts and cleaning up our proofs, and so it's fair to say that this paper in its current form owes a lot to Konstantin Genin. Thomas Icard's comments on a previous draft and our friendly interactions with him on related topics during our Stanford visits are very much appreciated. Finally, Kevin Kelly's work forms of course the basis and the inspiration for ours. Our frequent discussions with him in recent years influenced the development of our own perspective on the topic. He also gave us excellent reference tips concerning the history of the connections between topology and formal epistemology, as well as concerning his more recent work on related issues.

Nina Gierasimczuk's work on this paper was funded by an Innovational Research Incentives Scheme Veni grant 275-20-043, Netherlands Organisation for Scientific Research (NWO). Sonja Smets was funded in part by an Innovational Research Incentives Scheme Vidi grant from NWO, and by the European Research Council under the European Community's Seventh Framework Programme (FP7/2007-2013)/ERC Grant agreement no. 283963.

\nocite{*}
\bibliographystyle{eptcs}
\providecommand{\noopsort}[1]{}

\end{document}